\documentclass{LMCS}

\def\dOi{9(4:25)2013}
\lmcsheading%
{\dOi}
{1--18}
{}
{}
{Feb.~13, 2013}
{Dec.~25, 2013}
{}

\ACMCCS{[{\bf Theory of computation}]: Computational complexity and
  cryptography---Problems, reductions and completeness; Models of
  computation---Abstract machines}

\usepackage{hyperref}

\usepackage{color}

\newcommand{\bbbh}{\mathbb{H}}
\newcommand{\bbbr}{\mathbb{R}}
\newcommand{\bbbq}{\mathbb{Q}}
\newcommand{\bbba}{\mathbb{A}}

\newcommand{\bbbl}{\mathbb{L}} 
\newcommand{\bbbo}{\mathbb{O}} 
\newcommand{\bbbn}{\mathbb{N}} 
\newcommand{\bbbz}{\mathbb{Z}}
\newcommand{\bbbk}{\mathbb{K}}
\newcommand{\bbbp}{\mathbb{P}}

\definecolor{gruen}{rgb}{.0,0.4,0.1}

\begin{document}

\title[Strong Turing Degrees for Additive BSS RAM's]{Strong Turing Degrees for Additive BSS RAM's}

\author[Christine Ga\ss ner]{Christine Ga\ss ner}	
\address{Ernst-Moritz-Arndt-Universit\"at Greifswald, Germany}	
\email{gassnerc@uni-greifswald.de} 







\begin{abstract}
For the additive real BSS machines using only constants 0 and 1 and order tests we consider the corresponding Turing reducibility and characterize some semi-decidable decision problems over the reals. In order to refine, step-by-step, a linear hierarchy of Turing degrees with respect to this model, we define several halting problems for classes of additive machines with different abilities and construct further suitable decision problems. In the construction we use methods of the classical recursion theory as well as techniques for proving bounds resulting from algebraic properties.
In this way we extend a known hierarchy of problems below the halting problem for the additive machines using only equality tests
and we present a further subhierarchy of semi-decidable problems between the halting problems for the additive machines using only equality tests and using order tests, respectively. 
\end{abstract} 

\maketitle

\newcommand{\strpreceq}{\,\,{\preceq\!\!\!\!\!\!\!{\atop\not}}\,\,\,\,\,\,}
\section{Introduction} 
The real additive BSS RAM's considered here present a class of random access machines (RAM's) where each storage location can hold a real number. 
 We  want to  investigate   questions of decidability and  undecidability for this model of computation   and we hope that this research can also  help to provide answers to  questions  such as the following.

First, the additive BSS RAM is a special variant of a model of computation over an algebraic structure. We are convinced that the investigation of models of computation on a high algebraic level can help to understand the reasons for the non-computability or undecidability of certain algorithmic problems and to realize whether the model of computation, the algebraic structure  of the underlying data type, or the used intuitive algorithm is the cause of special undesired phenomena.
Using the additive machines, we are able to demonstrate the possibility to transfer and extend numerous results from the classical theory to the theory of computation over arbitrary structures. We believe that such transfers 
 can contribute to a better understanding of the universality of several methods and techniques developed for characterizing the classical notion of computation and, thus, that this study will contribute to answer questions such as the following (cf. \cite{Shepherdson}, p.\,\,582): \begin{quote}"What becomes of the concepts and results of elementary recursion theory if, instead of considering only computations on natural numbers, we consider computations on data objects from any relational structure?"\end{quote}\medskip

Second, for describing the practical efficiency of algorithms in theoretical terms, a higher level of abstraction 
can be useful in order to more adequately reflect the high-level programming languages. For that reason the use of real RAM's was proposed in \cite{Prepa},
 and today areas as the computational geometry cannot be imagined without various different types of such machines. 
 Depending on the problem, suitable operations of a real RAM (like addition, multiplication, the $k$-roots and so on) are chosen and are assumed as primitive operations that are available at unit cost. The first used models of this kind were additive or linear real RAM's that are able to execute additions and subtractions. They were introduced by D.\,Dobkin and R.\,J.\,Lipton, explicitly described in \cite{Dobk}, and are often used for determining 
better  upper and lower bounds for the time that is necessary to solve a problem by a certain algorithm in the worst case and giving the asymptotic worst-case complexity of problems. However, the calculated bounds of complexity of the problems result, in most cases, from the algebraic structure  of the underlying data type so that, for an intuitively given algorithm, one can find different upper and lower  bounds of complexity for different models (cf.\,\,\cite{Prepa}, p.\,\,29). 
 To address the question regarding  the difference between several models, it is important to give a characterization of the individual  models and to present typical features in detail. Since  questions on the complexity of problems can be closely related to the decidability of decision problems (cf.\,\,\cite{Oracles}), it appears reasonable, also, to
study the decidability and the reducibility of problems in the context of each single model.

Third, there is no doubt that it is possible to realize the addition of real numbers in special situations exactly. Only the comparison of two real quantities will remains a difficult problem. However, regardless of the current possibilities to solve problems over the reals exactly, we want to consider some questions: Which kind of problems and in which time such problems are exactly solvable if we had a machine being able to execute the addition and the subtraction of reals and to evaluate every order test in a finite (or in a fixed) time?

These issues also motivated us to deal with a question raised in \cite{6} and to define a hierarchy of semi-decidable problems above the set of rational numbers and below the Halting Problem for additive BSS machines over $\bbbr_{\rm add}^{1}=_{\rm df} (\bbbr;0,1;+,-;\geq)$.
 Here, we want to continue the study of semi-decidable problems over the reals in the framework of additive BSS RAM's that
fit best with the general model of computation considered e.\,g. in \cite{Oracles}. In particular, we want  to explore 
the following question: Is it possible to extend this hierarchy below the Halting Problem for additive machines over $\bbbr_{\rm add}^{1}$?

Whereas Turing machines only process a finite number of symbols and yield a uniform complexity theory, most register machines operate on a fixed number of integers or real numbers and provide non-uniform complexity classes. 
In proving a lower complexity bound for a problem, the size of the considered inputs is often an arbitrary fixed number, and in such cases the results are relevant for the uniform and the non-uniform complexity of the investigated problem. 
On the other hand, for determining the uniform complexity of an algorithm solving a simple problem such as the real Knapsack problem by means of the addition on real numbers we need a uniform model, for instance, an additive BSS model. The suitability of these machines follows from the fact that an input procedure yields the size of the input. It is typical for BSS machines that they can use this value in order to recognize the last input value in the infinite sequence of registers. Precisely this tiny extension implies that the BSS machines over $(\{0,1\};0,1;;=)$ can do the same as Turing machines, and more than most classical RAM's over $(\{0,1\};0,1;;=)$. 
Thus, a combination of properties of these models can lead to abstract machine models of computation over algebraic structures (BSS RAM's over algebraic structures) including the uniform BSS model of computation over the reals. In the setting of real BSS RAM's (whose definition has been inspired by classical random access machines and whose power agrees with one of the BSS machines), any real number can be considered as an algebraic object which can be stored in one register;
the permitted basic operations on the real numbers can be carried out within a fixed time unit; and moreover uniform algorithms that desire to process each number of real input values can be executed by a single machine. With regard to ability and efficiency, the additive BSS RAM's present a class of machines involving, for inputs over $\{0,1\}$, the BSS machines over $\{0,1\}$ and they are weaker than the BSS machines over the ring of reals (cf. \cite{1}). 
Moreover, the additive real BSS RAM's differ, like most RAM's, not only from the machines over $\{0,1\}$, but also from the classical RAM's over the integers (considered e.\,g.\,\,in \cite{AS79}). However, we will see that classical techniques can be used also to characterize the power of BSS RAM's.

Although all reducibility notions as for instance many-one reductions, truth table reductions or weak truth table reductions defined for Turing machines and studied in the classical recursion theory (cf. \cite{11}, \cite{10}) have an analogue in the theory of computation dealing with (additive) BSS RAM's over the real numbers, we will only consider the Turing reduction, i.\,e. the weakest of the mentioned reduction notions. This implies that non-reducibility results like the following are very general statements and remain valid for the other mentioned reduction notions, too. 
However, for structures over the reals it is useful to distinguish between a {\em weak} Turing reduction produced by oracle machines with real machine constants and a {\em strong} Turing reduction computed by means of the only constant 1. For both types of reduction we can consider the corresponding equivalence relation and the resulting 
 partition of decision problems over the reals into equivalence classes that we call the {\em weak} {\em Turing degrees} and the {\em strong} {\em Turing degrees}, respectively. It is not difficult to see that the partition into strong Turing degrees is
a finer refinement of the partition created by the weak Turing reduction. Also because of these properties, the strong Turing reduction 
will be the main subject in the following.

The papers \cite{6}, \cite{4}, and \cite{RMiller} present some hierarchies of semi-decidable problems below the halting problems for several types of BSS machines where most of the non-reducibility results follow from the algebraic properties of the real numbers.
We want to continue and deepen the study by investigating the strong reductions without arbitrary real parameters in order to even better understand the degrees of undecidability of decision problems over the reals.
Here, for additive machines over $\bbbr_{\rm add}^{1}$ we will extend an infinite hierarchy constructed by means of semi-decidable sets of tuples in $ \bbbr^n$ ($n\geq 1$), 
whose components satisfy a certain kind of linear dependence, in \cite{4}. For this purpose, we give the most important definitions in Section 2. In Section 3 we present a construction that is based on the enumerability of machines without irrational constants and their halting sets. 
Under these assumptions it is possible to transfer results from the classical recursion theory and use logical proof techniques. We will
apply 
 diagonalization techniques including the injury priority method that was used by Friedberg (in \cite{12}) and Muchnik in order to show that there are degrees between the Turing degrees $\emptyset$ and $\emptyset'$ of decidable and complete semi-decidable problems. 
 Section 3 offers a further subhierarchy between the halting problems for the additive machines using only the equality tests and using the order tests, respectively. In this construction we will again use algebraic and topological properties of the real numbers. In Section 4 we give a summary.

\section{The Model of Computation and the strong Turing reduction}
The uniform BSS model of computation was introduced in \cite{2} and the additive BSS machines were considered, for instance, 
in \cite{5}, \cite{3}, and \cite{6}. We will give a short definition of an additive machine in accordance with a model suitable for the computation over an arbitrary algebraic structure (cf. also \cite{Oracles}). The considered underlying algebraic structure over the real numbers will be the structure with addition and subtraction as basic operations and order as relation. Analogously to the general case where it is useful to distinguish between addresses of the registers and the individuals of the underlying structure and as introduced in \cite{18}, every machine ${ \mathcal M}$ is equipped with an infinite number of registers $Z_1,Z_2,\ldots$ for the real numbers and a finite number of index registers $I_1,I_2,\ldots, I_{k_{\mathcal M}}$.
The latter registers serve the uniform implementation of algorithms over arbitrary structures since they make it possible to process inputs of any length. First, they are used in copy instructions in analogy with those classical RAM's that are equipped with an accumulator (cf. \cite{Aho}, \cite{Mehlh}). Second, they obtain the size of every input tuple by an input procedure in analogy with the BSS model and they can be used in order to determine the length of the output.
Third, they allow to copy values without using individuals of the underlying structure  as  indices or addresses in copy instructions and -- in contrast to the model in \cite{2} -- to calculate all addresses without the help of the operations of the underlying structure.

The additive BSS RAM's can perform labeled instructions of the form $Z_i:= Z_{j}+ Z_{k}$, $Z_i:= Z_{j}- Z_{k}$, $Z_j:=c$, {\it if $ Z_{j}= 0$ then goto $l_1$ else goto 
$l_2$}, {\it if $ Z_{j}\geq 0$ then goto $l_1$ else goto 
$l_2$}, $Z_{I_j}:=Z_{I_k}$, $I_j:=1$, $I_j:=I_j+1$, and
{\it if $I_{j}=I_{k}$ then goto $l_1$ else goto $l_2$}. Each assignment of an input $(x_1,\ldots,x_n)\in \bbbr^\infty=_{\rm df}\bigcup_{i\geq 1}\bbbr^i $ to the registers of a machine ${ \mathcal M}$ can be realized by $Z_1:=x_1;\ldots;Z_n:=x_n; I_1:=n;\ldots;I_{k_{\mathcal M}}:=n$. We denote the class of all additive BSS machines by ${\sf M}_{\rm add}$. 
A problem $P\subseteq \bbbr^\infty$ is ${\sf M}_{\rm add}$-{\em semi-decidable} if its partial characteristic function is computable by a machine in ${\sf M}_{\rm add}$, and a set is ${\sf M}_{\rm add}$-{\em decidable} if its characteristic function is computable by a machine in ${\sf M}_{\rm add}$. Thus the {\em halting sets} of machines in ${\sf M}_{\rm add}$ are ${\sf M}_{\rm add}$-semi-decidable, and a set is ${\sf M}_{\rm add}$-decidable if and only if the set itself and its complement are ${\sf M}_{\rm add}$-semi-decidable.
 A set $S\subseteq \bbbr^\infty$ is {\em recursively} (or {\em effectively}) {\em enumerable by a machine in ${\sf M}_{\rm add}$} (or {\em ${\sf M}_{\rm add}$-enumerable}) 
 if a surjective function $f:\bbbn \to S$ is computable by an ${\sf M}_{\rm add}$-machine. Note that this definition of the notion {\em effectively enumerable set} differs from the definition of recursively enumerable sets given in \cite{2}. 

A significant difference between the Turing model and the BSS (RAM) model becomes clear by considering the relation between semi-decidability and enumerability.
In the classical theory of recursive functions and in the theory of recursive reducibility, we mainly speak about recursive and recursively enumerable sets. 
The {\em recursive sets} $S\subseteq \bbbn$ are decidable by Turing machines and 
 the {\em recursively enumerable} problems $S\subseteq \bbbn$ are the sets for which a bijective function $f:\bbbn \to S$ is computable by a Turing machine. The sets semi-decidable by Turing machines are the halting sets of Turing machines -- and consequently reducible to the Halting Problem for Turing machines with respect to many-one reducibility -- and they are recursively enumerable by Turing machines. Thus, a set of positive integers is decidable by a Turing machine  if and only if the set itself and its complement are effectively enumerable.
The halting sets of machines in ${\sf M}_{\rm add}$ can also be reduced to the Halting Problem for machines in ${\sf M}_{\rm add}$. 
 However, there are halting sets $S\subseteq \bbbr^{\infty}$ (e.\,g. $S=\bbbr$) which are not effectively ${\sf M}_{\rm add}$-enumerable.

Let ${\sf M}_{\rm add}^{1,=}$ and ${\sf M}_{\rm add}^1$ be the classes of the additive BSS machines  over $(\bbbr;0,1;+,-;=)$ and $\bbbr_{\rm add}^{1}$, respectively, using only the constants 0 and 1 in instructions of the form $Z_j:=c$ and only the equality tests and order tests, respectively. Let notions such as {\em ${\sf M}_{\rm add}^{1,=}$-decidable} be defined analogously to the definitions given above.

The subsets of $\bbbr$ semi-decidable or decidable by additive machines are easy to characterize (cf.\,\,Proposition 1 in \cite{2}, where the authors call all halting sets {\em recursively enumerable}) since, 
for any input $x\in \bbbr$, at any time of the computation the content of any register of an additive machine using only constants 1 and 0
can be described by a term $kx+l$ with $k,l\in \bbbz$ and, according to this,
 tests along a computation path correspond to questions of the form {\em $kx+l=0$?} and {\em $kx+l\geq 0$?}, respectively.
 The class of $ {\sf M}_{\rm add}^{1,=}$-semi-decidable problems $P\subseteq \bbbr$ contains  countable sets $\subseteq \bbbq$ and the co-finite sets $S$ with $S\supseteq \bbbr\setminus \bbbq$,
 the class of $ {\sf M}_{\rm add}^{1,=}$-decidable problems $P\subseteq \bbbr$ contains the finite sets $\subseteq \bbbq$ and the co-finite sets $S$ with $S\supseteq \bbbr\setminus \bbbq$,
 the problems $P\subseteq \bbbr$ semi-decidable by a machine in $ {\sf M} _{\rm add}^{1}$ are
 countable unions of intervals, and,
for  any $ {\sf M} _{\rm add}^{1}$-decidable problem $P\subseteq \bbbr$, the sets $P$ and $\bbbr\setminus P$  are countable unions of disjoint intervals  with end points  in $\bbbq\cup \{-\infty,\infty\}$.
The latter holds since any set $S$ of inputs for which an additive machine in $ {\sf M} _{\rm add}^{1}$ follows a computation path is convex and, consequently, an interval if $S\subseteq \bbbr$. If the machine decides a problem, then each of these computation paths  is finite and, thus, there are only a countable number of paths. In case of machines over $\bbbr _{\rm add}^{1}$ and inputs from $\bbbr$, we can order the paths such that a path $B_1$ is ''left'' from another path $B_2$ if and only if the inputs for which a machine goes through $B_1$ are less than the inputs for which the machine goes through $B_2$. Because we assume that any computation path is either an accepting or a rejecting path, any path corresponds to an interval of accepted inputs or to one of rejected inputs. The $ {\sf M} _{\rm add}^{1}$-semi-decidable subsets of $\bbbr$ are countable unions of disjoint intervals with limits of computable increasing or decreasing sequences as end points. 

By analogy with the classical Turing reduction, we want to consider reductions by means of oracle machines that possess an oracle ${ \mathcal O} \subseteq \bbbr^\infty$ and are able to make tests whether a given tuple $\vec x \in \bbbr^\infty$ belongs to ${ \mathcal O}$ or not. To realize the evaluation of queries of the form {\em $\vec x \in { \mathcal O}$?} we allow to use oracle instructions of the form\begin{equation}\label{OracleInstruction}\mbox{\em if $(Z_1,\ldots,Z_{I_1})\in \!{ \mathcal O}$ then goto $l_1$ else goto $l_2$}\end{equation}
where the length of any tuple in a query is determined by the content of the index register $I_1$. 
Note that we prefer to use this form of oracle instructions explicitly introduced in \cite{18} in order to take into account the possibility to use the same form also in a general model of computation over arbitrary structures. Accordingly, let ${\sf M}_{\rm add}^1({ \mathcal O})$ be the class of the additive machines over $\bbbr_{\rm add}^1 $ which can additionally execute instructions of the form (\ref{OracleInstruction}).
 Like in \cite{6}, we say that a problem ${ \mathcal A}\subseteq \bbbr^\infty $ is {\em easier than} a problem ${ \mathcal B}\subseteq\bbbr^\infty$ (denoted by ${ \mathcal A}\preceq_{\rm add}^{1}{ \mathcal B}$) 
if ${ \mathcal A}$ is decidable by a machine in ${\sf M}_{\rm add}^{1} ({ \mathcal B})$. If ${ \mathcal A}$ is easier than ${ \mathcal B}$, then ${ \mathcal B}$ is {\em harder than} ${ \mathcal A}$. The problem ${ \mathcal A}$ is {\em strictly easier than} 
${ \mathcal B}$ (denoted by ${ \mathcal A} \strpreceq\!\!_{\rm add}^{1}\,\,{ \mathcal B}$) if ${ \mathcal A}\preceq_{\rm add}^{1}{ \mathcal B}$ holds and ${ \mathcal B}$
cannot be decided by a machine in ${\sf M}_{\rm add}^{1} ({ \mathcal A})$. We say that ${ \mathcal A}$ can be {\em Turing reduced to ${ \mathcal B}$ by a machine in ${\sf M}_{\rm add}^{1}$} -- or we could also say that ${ \mathcal A}$ can be {\em BSS reduced to ${ \mathcal B}$ by a machine in ${\sf M}_{\rm add}^{1}$} -- if ${ \mathcal A}$ is easier than ${ \mathcal B}$. Let ${\mathcal A}\equiv_{\rm add}^{1}{ \mathcal B}$ if and only if ${\mathcal A}\preceq_{\rm add}^{1}{\mathcal B}$ and ${\mathcal B}\preceq_{\rm add}^{1}{\mathcal A}$. Thus, for all decision problems ${ \mathcal A},{ \mathcal B}\subseteq \bbbr^\infty$ and $ { \mathcal A}^{\rm c}\,\,=_{\rm df} \,\,\bbbr^\infty \setminus { \mathcal A}$, we have ${ \mathcal A} \equiv_{\rm add}^{1} { \mathcal A}^{\rm c}$, ${ \mathcal A}\preceq_{\rm add}^{1} \{1\}\times { \mathcal A}\cup \{2\}\times { \mathcal B}$, and if ${ \mathcal A}$ and ${ \mathcal B}$ are disjoint and ${\sf M}_{\rm add}^{1}$-semi-decidable, then we have ${ \mathcal A}\preceq_{\rm add}^{1} { \mathcal A}\cup { \mathcal B} $. ${\sf M}_{\rm add}^{1,=}({ \mathcal O})$  and  $\preceq_{\rm add}^{1,=}$ can similarly be defined.
 To simplify matters, we will write ${ \mathcal A}\preceq{ \mathcal B}$, 
 ${ \mathcal A} \strpreceq{ \mathcal B}$, and ${ \mathcal A} \equiv { \mathcal B}$ for and only for
${\mathcal A}\preceq_{\rm add}^{1}{ \mathcal B}$, ${ \mathcal A} \strpreceq\!\!_{\rm add}^{1}\,\,{ \mathcal B}$, and ${\mathcal A}\equiv_{\rm add}^{1}{ \mathcal B}$, respectively.

 For the classes ${\sf M}_{\rm add}^{1} $, ${\sf M}_{\rm add}^{1,=} $, and ${\sf M}_{\rm add} $ we will now consider the 
halting problems defined by 
\[\bbbh_{\rm add}^{[1[,=]]} = \bigcup_{n\geq 1} \{(n\,.\,\vec x\,.\,{\rm code}({ \mathcal M}))\mid {\vec x}\in \bbbr ^n \,\,\&\,\, { \mathcal M}\in{\sf M}_{\rm add}^{[1[,=]]} \,\,\&\,\, { \mathcal M}(\vec x)\downarrow \}\] 
where we write $(n\,.\,\vec x\,.\,{\rm code}({ \mathcal M}))$ instead of $(n,x_1,\ldots,x_n,s_1,\ldots,s_m)$ if $\vec x= (x_1,\ldots,x_n)$ and ${\rm code}({ \mathcal M})=( s_1,\ldots,s_m)$. ${ \mathcal M}(\vec x)\downarrow $ and ${ \mathcal M}(\vec x)\uparrow $ will be used to express the facts that ${ \mathcal M}$ halts on $\vec x$ or not. 
The codes ${\rm code}({ \mathcal M})$ of the machines ${ \mathcal M}\in{\sf M}_{\rm add}^{1}$ are sequences of the codes of the single symbols of their programs (in analogy with \cite{2}, p. 34) and, for some $k$, each of the single symbols is encoded by $k$ symbols in $\{0,1\}$ unambiguously (in analogy with the classical setting where Turing machines are considered). 
For every other BSS machine ${ \mathcal M} \in {\sf M}_{\rm add}$, let ${\rm code}({ \mathcal M})$ be the usual code of its program where any real constant is encoded by itself and the other single symbols are encoded by $k$ zeros and ones (for some $k$). Consequently, it is easy to see that we have
\[ \bbbh_{\rm add}^{1,=} \preceq \bbbh_{\rm add}^{1} \equiv \bbbh_{\rm add}.\] The latter is valid since $\bbbh_{\rm add}$ is semi-decidable by  a machine in ${\sf M}_{\rm add}^1$.

 Since $\bbbn=_{\rm df } \{0,1,\ldots\}$ is effectively enumerable by a machine in ${\sf M}_{\rm add}^{1,=}$, $\bbbn$ and $\bbbn_+=_{\rm df } \{1,2, \ldots\}$ are   ${\sf M}_{\rm add}^{1}$-decidable and    ${\sf M}_{\rm add}^{1,=}$-semi-decidable. Therefore,  with respect to the machines in  ${\sf M}_{\rm add}^{1,=}$, we  have $\bbbn
\preceq _{\rm add}^{1,=} \bbbh_{\rm add}^{1,=}$ which means that   $\bbbn$  and $\{(x,y)\in \bbbn^2\mid x\geq y\}$ are  decidable by  machines in ${\sf M}_{\rm add}^{1,=}(\bbbh_{\rm add}^{1,=})$.
If an input $\vec x$ is in $\bbbn^\infty$, then a machine in ${\sf M}_{\rm add}^{1,=}$ can determine  the binary representation of $\vec x$ and simulate a given Turing machine. On the other hand, any addition of integers and, for any $x,y\in \bbbn$, the  test $x\geq y${\em ?} can be simulated by Turing machines. Therefore,     the Halting problem for Turing machines is  ${\sf M}_{\rm add}^{1,=}$-semi-decidable and    $\bbbh_{\rm add}^{1,=}\cap \bbbn^\infty$ is semi-decidable by a     Turing machine.       
Moreover, every problem $P\subseteq \bbbn$ is semi-decidable by a Turing machine if and only if it is ${\sf M}_{\rm add}^{1,=}$-semi-decidable  and 
the problem is decidable by a Turing machine
 if and only if it is ${\sf M}_{\rm add}^{1}$-decidable. 
We have also the following.

\begin{prop}\label{Turingred} For all $P\subseteq \bbbn_+$, $P$ is decidable by an oracle Turing machine using the Halting Problem for Turing machines as oracle if and only if $P \preceq_{\rm add}^{1,=}\bbbh_{\rm add}^{1,=} $ and  if and only if $P \preceq\bbbh_{\rm add}^{1} $. 
\end{prop}

\begin{cor}We have
 $(\bbbh_{\rm add}^{1,=}\cap \bbbn^\infty)\,\, \equiv_{\rm add}^{1,=}\,\, (\bbbh_{\rm add}^{1} \cap \bbbn^\infty)\,\, \preceq_{\rm add}^{1,=} \,\,\bbbh_{\rm add}^1 $. \end{cor} 

 For any problem $Q\subseteq \bbbn_+$  below the Halting Problem for Turing machines,  there is no guarantee that    $\bbbn \preceq_{\rm add}^{1,=}Q$ holds. For this reason we restrict the  following statement to the relation $\preceq $.

\begin{prop}For all $P,Q\subseteq \bbbn_+$,  $P$ is decidable by an oracle Turing machine using $Q$ as oracle   if and only if     $P \preceq Q$.
\end{prop}

That means that the classical Turing degrees 
 form a substructure of the structure of the degrees with respect to the machines in ${\sf M}_{\rm add}^{1}$.  However, no classical degree is 
equal to the corresponding degree with respect to $\preceq $. Among others, we have, for instance, $\{\frac{n}{k}\mid n\in \bbbz\} \equiv \bbbn\equiv \emptyset$ for any $k\in \bbbn_+$. 
 Moreover, we are able to show that $\{\frac{n}{k}\in \bbbq\setminus \bbbz\mid n\in P\}\strpreceq P$ holds for any $k\in \bbbn_+$ and any undecidable $P\subseteq \bbbn_+$. The latter means that there are further degrees.

\begin{cor} 
The class of all Turing degrees resulting from the classical Turing reduction represent a proper subclass of Turing degrees defined by the relation $\equiv$.
\end{cor}

Since additive machines are able to evaluate several properties of rational as well as irrational  inputs, there are further real Turing degrees such as the degrees represented by  $\bbbl_1, \bbbl_2,\ldots$   in (\ref{ErsteHierarchie}).  
The answer of the question about the decidability of these  problems is dependent on  the algebraic  properties  of the underlying structure and we will  show that  
  the ${\sf M}_{\rm add}^{1}$-semi-decidable  problems  $\subseteq \bbbn_+$ and   $\bbbl_1, \bbbl_2,\ldots$ are incomparable with respect to the relation  $\preceq$.  In \cite{4} we showed
 \begin{equation}\bbbq= \bbbl_1\mbox{$\strpreceq$} \bbbl_2\mbox{$\strpreceq$} \cdots \mbox{$\strpreceq$} \bbbl_k\mbox{$\strpreceq$} \cdots\mbox{$\strpreceq$} {\bbbl \preceq \bbbh_{\rm add}^{1}}\label{ErsteHierarchie}\end{equation} where 
$ \bbbl _n=\{ (x_1,\ldots,x_n) \in \bbbr^n
\mid (\exists (q_0, \ldots,q_{n-1})\in \bbbq ^{n} ) (q_0+\sum_{i=1}^{n-1}q_ix_i =x_n)\}$ and $ \bbbl = \mbox{$\bigcup_{k\geq 1}$}\bbbl_k$ are subsets of $\bbbr^\infty $ whose elements can be described by certain linear equations. And a closer look reveals that $\bbbl$ is even easier than $\bbbh_{\rm add}^{1,=}$. Consequently, this raises the question whether $\bbbh_{\rm add}^{1,=}$ and $\bbbh_{\rm add}^{1}$ define the same Turing degree with respect to $\equiv$, and so on. Therefore, we shall now extend our results 
 to a refinement of this hierarchy by presenting a subhierarchy between $\bbbh_{\rm add}^{1,=}$ and $\bbbh_{\rm add}^{1}$ and constructing decision problems that are easier than $\bbbh_{\rm add}^{1,=}$ and harder the $\bbbl$.

\section{The construction of a further problem below $\bbbh_{\rm add}^{1,=}$ by a priority method }

By properties which we recalled above and classical results there are various semi-decidable problems $P\subseteq \bbbr$ and among them some  halting problems $P\subseteq \bbbn$ that are all not decidable by an additive machine without irrational constants. 
Here, we construct an ${\sf M}_{\rm add}^{1,=}$-semi-decidable problem $\bbba\subseteq \bbbn$ below $\bbbh_{\rm add}^{1,=}$ by means of a priority method often used in the classical recursion theory for constructing recursively enumerable sets (cf. \cite{9}, \cite{10}, and \cite{8})  and then  we  extend the results.
 This construction serves two goals. On one hand, we want to contribute to a classification of the ${\sf M}_{\rm add}^{1}$-semi-decidable problems, and on the other hand we want 
to take  the question  into account whether it is possible to transfer results of elementary recursion theory. 
We  will show that the transfer of ideas is possible if
the proofs are carefully analyzed and semantic and syntactic techniques that are usually combined become clearly separated. 
In particular, semantic techniques  can, in our experience, be successfully applied to discuss fundamental questions on the power of random access machines for several underlying structures. This is also confirmed by the following    application 
of the priority  method  that  provides a general way of constructing
  semi-decidable problems  below the halting problem over several structures  of finite signature.  Constructions such as the following are also possible  in cases where all problems that are  semi-decidable with respect to Turing machines are decidable by machines over the considered algebraic  structure (see also the remark after Proposition \ref{AnichtaufK}). 

We want to  define the set $\bbba$ recursively such that $\bbba$ has  the following properties. 
\begin{enumerate}
 \item $\bbba $ is ${\sf M}_{\rm add}^{1,=}$-enumerable and thus ${\sf M}_{\rm add}^{1,=}$-semi-decidable.

 \item $\bbbn\setminus \bbba$ is infinite.

\item
The intersection of $\bbba$ and each infinite ${\sf M}_{\rm add}^{1,=}$-enumerable subset of $\bbbn$ is not empty.

 \item The halting problem $\bbbh_{\rm add}^{1,=}$ cannot be reduced to $\bbba$ by a machine in ${\sf M}_{\rm add}^{1[,=]}$.
\end{enumerate}
Thus, by (2) and (3), 
 $\bbba$ is not decidable by a machine in ${\sf M}_{\rm add}^{1,=}$.
Moreover, we will see that this statement  also remains true if, for   deciding $\bbba$,  we allow
  all additive oracle machines in ${\sf M}_{\rm add}^{1}({ \mathcal Q})$ 
for which  the problem ${\mathcal Q}\cap \bbbz^\infty$    is ${\sf M}_{\rm add}^1$-decidable. The reason is that, for inputs in $\bbbz^\infty$, the oracle queries $x\in { \mathcal Q}${\em ?} have no benefit since they can be simulated by machines in ${\sf M}_{\rm add}^1$ without oracle queries. 
That means that we will, in such a situation, get $\bbba \not\preceq { \mathcal Q}$.

 Provided that a structure of finite signature such as $\bbbr_{\rm add}^{1,=}$ contains a set $\bar \bbbn$ that is decidable and effectively enumerable over this structure, then it is possible to construct a set $\bar \bbba\subseteq \bar \bbbn$ with properties corresponding to (1), \ldots, (4).
 Other algebraic properties of the underlying structure are irrelevant. From this point of view, semi-decidable sets such as $\bar \bbba$ are very simple in the double sense of the word. By analogy with the classical theory of computation we also want to say that a set is {\em simple} ({\em with respect to the underlying structure}) if it is effectively enumerable and its complement related to $\bar\bbbn$ is infinite and contains no infinite effectively  enumerable set, over this structure. In this context, the simplicity of a set that follows, for $\bbba$, from (1), (2), and (3) implies its undecidability. The property (4) follows from the lowness of $\bbba$ where we say that an effectively enumerable set $S$ is {\em low } ({\em with respect to the underlying structure}) if $\bbbk^S \preceq \bbbk^\emptyset$ holds for special halting problems $\bbbk^S$ and $\bbbk^\emptyset$ that we want to define now.

 The fact that the codes of the machines in ${\sf M}_{\rm add}^{1}({ \mathcal O})$ as well as in
${\sf M}_{\rm add}^{1,=}$ can be effectively enumerated by additive machines in ${\sf M}_{\rm add}^{1,=}$ and the next lemma will be of particular importance in the following.

\begin{lem} {\it The intersection of the halting set of any machine in $\,{\sf M}_{\rm add}^{1,=}$ and $\bbbn$ is effectively enumerable and, consequently, it is also a halting set of a machine in $\,{\sf M}_{\rm add}^{1,=}$. }\end{lem}

\begin{proof} Let $H_{ \mathcal M}$ be the halting set of a machine ${ \mathcal M}\in {\sf M}_{\rm add}^{1,=}$. Then, a machine in ${\sf M}_{\rm add}^{1,=}$ can enumerate all pairs $(n_1,t_1),(n_2,t_2),\ldots$ $\in \bbbn^2$ for which ${ \mathcal M}$ halts on $n_i$ after exactly $t_i$ steps.
Thus, there is also an enumeration machine $\bar{ \mathcal M}$ $\in {\sf M}_{\rm add}^{1,=}$ computing,
 for the input $j$, the first pair $(n_1,t_1)$ from 1 and each further pair $(n_{i+1},t_{i+1})$ step-by-step from its enumerated predecessor $(n_{i},t_{i})$ for all $i<j$, and outputting the integer $n_j$. Consequently, $H_{ \mathcal M}\cap \bbbn$ is semi-decidable by a machine in ${\sf M}_{\rm add}^{1,=}$ using the program of $\bar{ \mathcal M}$ as subprogram.
\end{proof}

For any ${ \mathcal M}\in {\sf M}_{\rm add}^1({ \mathcal O})$, let $ K_{ \mathcal M}=2^{|{\rm code({ \mathcal M}}) |}+ c_{ \mathcal M}$, where $c_{ \mathcal M}$ is the integer whose binary code matches essentially (apart from leading zeros) with ${\rm code({ \mathcal M}})\in\{0,1\}^\infty$ (that means, more precisely, that $ K_{ \mathcal M}= \sum_{i=0}^lc_{i} 2^{l-i}$
if $c_0=1$ and ${\rm code}({ \mathcal M})=(c_1,\ldots,c_l)$) 
 and let ${ \mathcal M}_{ K_{ \mathcal M}} ^{ \mathcal O}={\mathcal M} $. If a number $i$ is not in $ \{ K_{ \mathcal M} \mid { \mathcal M}\in {\sf M}_{\rm add}^1({ \mathcal O}) \}$, then let the corresponding machine ${ \mathcal M}_i ^{ \mathcal O} $ be a simple machine that halts in any case. 
In this way, for any ${ \mathcal O}\subseteq \bbbn$, we get a list of all additive oracle machines by
${ \mathcal M}_1^{ \mathcal O}, { \mathcal M}_2^{ \mathcal O}, \ldots$.
Let ${ \mathcal N}_1, { \mathcal N}_2,\ldots$ be a list of all
additive machines in ${\sf M}_{\rm add}^{1,=}$ and let $\bar {
  \mathcal N}_i$ enumerate the set $W_i$ of all positive integers $n_{i,1},n_{i,2},\ldots$ that belong to the halting set of ${ \mathcal N}_i$.
 For any oracle set ${ \mathcal O}\subseteq \bbbr^\infty $ we will  use the following special halting problem.
\begin{displaymath}\begin{array}{c} 
\bbbk^{ \mathcal O} = \bbbh_{\rm spec}({\sf M}_{\rm add}^{1}( { \mathcal O}) ) \,\,=_{\rm df} \,\,\{ k\in \bbbn_+ \mid { \mathcal M}_k^{ \mathcal O}  \in {\sf M}_{\rm add}^{1}( { \mathcal O}) \,\,\&\,\, { \mathcal M}_k^{ \mathcal O} (k)\downarrow \}.
\end{array}\end{displaymath} Then, both following relations are trivial since, on integers, each order test can be simulated by means of equality tests.

\begin{lem} $ \bbbk^\emptyset \preceq \bbbh_{\rm add}^{1,=}$ and $\bbbk^\emptyset \preceq \bbbh_{\rm add}^1$.\end{lem}

By using the ideas of the proof of Proposition 6 in \cite{Oracles} we get the following.

\begin{lem} $ \bbbk^\emptyset \equiv \bbbh_{\rm add}^{1,=}\cap \bbbn^\infty\equiv \bbbh_{\rm add}^{1}\cap \bbbn^\infty$.\end{lem} 

 By Proposition \ref{Turingred}, the following corollary states that the power of $\bbbk^\emptyset$ in the class of decision problems $P\subseteq \bbbn^\infty$ with respect to the relation $ \preceq$ (that means to $ \preceq_{\rm add}^1$) corresponds to one of the Halting Problem for Turing machines within the classical theory. This is easy to show since $\bbbn$ is ${\sf M}_{\rm add}^{1}$-decidable and, consequently, we can use the relationship $\bbbn \preceq \bbbk^\emptyset$. But, with respect to the relation $ \preceq_{\rm add}^{1,=}$ the power of $\bbbk^\emptyset$ is weaker; we have,  for any ${\mathcal O}\subseteq \bbbr^\infty$,  $\bbbn \not\preceq _{\rm add}^{1,=} \bbbk^{\mathcal O}$. The latter can be shown as follows. Since $\bbbk^{\mathcal O}$ is not ${\sf M}_{\rm add}^{1}$-decidable (cf.\,\,Lemma \ref{unentOrakelm}), $\bbbn \setminus \bbbk^{\mathcal O}$ is not finite. That means that any oracle machine in ${\sf M}_{\rm add}^{1,=} (\bbbk^{\mathcal O})$ goes, for all irrational inputs and for an infinite number of positive integers, through the finite computation path that results from answering nontrivial tests {\em $kx+l=0$?} (where $k\not=0$) and queries {\em $x\in \bbbk^{\mathcal O}$?} in the negative.

\begin{cor}\label{ReduktAequv} For any $P\subseteq \bbbn$, $P \preceq \bbbk^\emptyset$ if and only if $P \preceq \bbbh_{\rm add}^{1,=}$  and if and only if $P \preceq \bbbh_{\rm add}^1$.\end{cor}

In analogy with the classical setting, we have also the following.
\begin{lem}\label{unentOrakelm} For any oracle set ${ \mathcal O}\subseteq \bbbr^\infty $, we have $ \bbbk^{ \mathcal O} \not \preceq { \mathcal O}$.\end{lem}

Now, the set $\bbba=\bigcup_{s\geq 1}\bbba_s$ will be defined in stages where, for any machine ${ \mathcal M}$ and any input $\vec x$, ${ \mathcal M}(\vec x)\downarrow ^t$ means that ${ \mathcal M}$ halts on $\vec x$ within $t$ steps and ${ \mathcal M}(\vec x)\uparrow ^t$ means that ${ \mathcal M}$ does not halt on $\vec x$ within $t$ steps. 
 Let
 $\bbba_1=\emptyset$. Assume that, for $s\geq 1$, $\bbba_s$ is already defined. Then, for any $j\leq s$,
let {$a(j,s)$} be defined as follows.
 If ${ \mathcal M}_j^{\bbba_s}(j)\downarrow ^s$, then let
{$a(j,s)$} 
be the greatest integer which is used in a query by ${ \mathcal M}_j^{\bbba_s}$ on input $j$ within the first $s$ steps, 
and if ${ \mathcal M}_j^{\bbba_s}(j)\uparrow ^s$, then let
$a(j,s)$ be 0. Moreover, let $W_{i,s} \subseteq W_i$ 
be the set  of  the positive integers enumerated by $\bar { \mathcal N}_i$ for the input $s$ within the first $s$ steps and let
\[I_{s}= \{i\leq s\mid \bbba_s \cap  W_{i,s} =\emptyset \,\,\&\,\, (\exists x \in W_{i,s})\phi(i,s,x) \}\]
where $\phi(i,s,x) =_{\rm df} 2i<x \,\,\&\,\, (\forall j\leq i) (a(j,s)< x)$. If $I_{s}\not=\emptyset$, then let $i_s=\min I_{s}$ be the {\em active index} for extending $\bbba_s$ by
 $x_{i_s}=\min \{x \in W_{i_s,s} \mid\phi( i_s,s,x) \}$. Finally let
\begin{displaymath} \begin{array}{lcl}\label{DefA}
\bbba_{s+1} &=_{\rm df}& 
\left\{\begin{array}{l@{\quad}l}\bbba_{s} 
 &\mbox{if $I_{s}=\emptyset$},\\ 
 \bbba_{s}\cup \{x_{i_s}\} &\mbox{otherwise.}\\ \end{array} \right.
\end{array}\end{displaymath}
 \begin{lem}\label{ErsteEigenschaft}
 $\bbba $ is effectively enumerable by an additive machine in ${\sf M}_{\rm add}^{1,=}$.
\end{lem}
\begin{proof}
For $s\geq 1$,  let us assume that  $\bbba_s$ is ${\sf M}_{\rm add}^{1,=}$-enumerable  and let us  consider 
\begin{displaymath}\begin{array}{lll} 
W_{1,s}&=&\{n_{1,1},n_{1,2},\ldots, n_{1,t_1}\},\\
W_{2,s}&=&\{n_{2,1},n_{2,2},\ldots, n_{2,t_2}\},\\
 &\ldots,\\
W_{s,s}&=&\{n_{s,1},n_{s,2},\ldots, n_{s,t_s}\} \\
\end{array}\end{displaymath}
where $t_1,\ldots,t_s<s$.
 Then, for any $i=1,\ldots,s$ and $k=1,\ldots,t_i$, the inequalities 
$\phi(i,s,n_{i,k})$
 can be checked by an additive machine in ${\sf M}_{\rm add}^{1,=}$ since each ${ \mathcal M}_j^{\bbba_s}$ computes only integers on input $j$ and, for all  integers $x$, any query $x\in \bbba_s${\em ?} and   any test $x\geq 0${\em ?} are decidable by a machine in ${\sf M}_{\rm add}^{1,=}$. 
If we have the inequalities $\phi(i',s,n_{i',k})$ for some $i' \leq s$ with  $  \bbba_s  \cap W_{i',s}=\emptyset$,
 then $I_s\not=\emptyset$ and we can fix $i_s$ and enumerate the next $x_{i_s}\in \bbba_{s+1}\setminus \bbba_s$. Note that, in the following steps, the index ${i_s}$ will 
 not be considered further since we have $\bbba_{s+1}\cap  W_{i_s,s} =\{x_{i_s}\}$. 
\end{proof}
 \begin{lem}\label{ZweiteEigenschaft}
 $\bbbn\setminus \bbba$ is infinite. 
\end{lem}
\begin{proof} By definition of $\bbba$, there is a sequence of active indices $i_{m_1}< i_{m_2}< \cdots$ such that $\bbba=\{x_{i_{m_1}},x_{i_{m_2}},\ldots \}$ and $2i_{m_r} <x_{i_{m_s}}$ for all $r\leq s$. Therefore, we have 
 $|\bbba \cap \{0,\ldots, 2 i\}|< i$ and thus $|\{0,\ldots, 2 i\}\setminus \bbba|> i$
for all $ i$.
\end{proof}

In the following we will need that the intersection of $\bbba$ and any
infinite halting set $W_i\subseteq \bbbn$ is not empty. Let us show it
in two stages by proving the conditions for lowness and simplicity for
$\bbba$.

\begin{lem}[Conditions for lowness] For all $n>0$, $(N_n)$ holds.

\vspace{0.2cm}

 {($N_n$)} If ${ \mathcal M}_n^{\bbba_t}(n)\downarrow^t$ for infinitely many $t$, then ${ \mathcal M}_n^\bbba(n)\downarrow$.
\end{lem}
\begin{proof}
 Let us assume that, for an $n=j$, there is an infinite sequence $t_1< t_2< \cdots$ such that 
\begin{equation}\label{Folge00} { \mathcal M}_j^{\bbba_{t_1}}(j)\downarrow^{t_1}, { \mathcal M}_j^{\bbba_{t_2}}(j)\downarrow^{t_2}, \ldots.\end{equation}
 Let $s$ be great enough such that, for each $t\geq s$, either
$I_t=\emptyset \mbox{ or } j \leq i_t=\min I_t$. The existence of $s$ is ensured since the active indices form a sequence
 $(i_t)_{t\geq 1}$ without repetition.
 It follows from
(\ref{Folge00}) that there is at least one $t_\nu \geq s$ such that ${ \mathcal M}_j^{\bbba_{t_\nu}}(j)\downarrow^{t_\nu}$. Consequently, for any $t=t_\nu,t_\nu+1,\ldots $ for which $I_t\not=\emptyset$ and $x_{i_t}$ is defined, we have $x_{i_t}> a(j,t_\nu) $ and $ a(j,t_\nu)=a(j,t_\nu+1)=\cdots \,\,$ since the set $\bbba\setminus \bbba_{t_\nu}$ contains, by definition, no elements appearing in a query of ${\mathcal M}_j^{\mathcal O}$ on $j$ if ${\mathcal O}={{\bbba}_{t_\nu+1}},{{\bbba}_{t_\nu+2}}, \ldots$. Therefore, 
 the computation of ${\mathcal M}_j^{\bbba}$  on $j$
proceeds in exactly the same
manner as one of 
${ \mathcal M}_j^{\bbba_{t_\nu+1}}$, ${ \mathcal M}_j^{\bbba_{t_\nu+2}}, \ldots$, receives the same answer to its oracle queries, and halts with
the same output in the same number ($\leq t_\nu$) of steps.
 \end{proof}

\begin{lem}[Conditions for simplicity] \label{ConditionsSimplicity}

For all $n>0$, $(P_n)$ holds.

\vspace{0.2cm}
{$(P_n)$} If $W_n=\bigcup_{i\geq 1} W_{n,i}$ is infinite, then $\bbba\cap W_n\not=\emptyset$.
\end{lem}

\begin{proof}
 Let us assume that   $(P_n)$ does not hold for an   $n>0$. Then, let    $i$ be the smallest index such that the
 property $(P_i) $ fails. Consequently, $W_i$ is infinite and  $\bbba \cap W_i =\emptyset $.  Moreover, for any $i' < i$, let $r_{i'}$ be a positive integer satisfying $\bbba\cap W_{i',r_{i'}}  \not=\emptyset$ if $\bbba\cap W_{i'}\not =\emptyset$ --
regardless of whether
$W_{i'}=\bigcup_{s\geq 1} W_{i',s} $ is a finite or an infinite set. Moreover, let $r_i$ be a positive integer with $W_{i,r_{i}}\setminus \{0,\ldots, 2i\} \not=\emptyset$.
Then, by definition of $\bbba$ and because of $\bbba\cap W_{i} =\emptyset$, there is some $s_0\geq \max \{r_{i'}\mid i'\leq i\}$ such that for any $s\geq s_0$ and any $x\in W_{i,s}$ with $x> 2i$, there is some $j_s\leq i$ such that $a(j_s,s)\geq x$.
Consequently, there is an element $k$ in the subset $\bigcup_{s\geq s_0}\{j_s\}$ of the finite set $\{ 1,\ldots, i\}$ such that $k=j_s$ holds for infinitely many $s\geq s_0$ and 
\begin{equation}
a( k,s_1)\geq \max (W_{i,s_1}),\quad a( k,s_2)\geq \max (W_{i,s_2}), \quad \ldots \label{folge1}\end{equation} hold
for an infinite sequence $s_1, s_2,\ldots$ with increasing maxima $\max (W_{i,s_1}) < \max (W_{i,s_2}) < \cdots$.
The maxima of $W_{i,s_1},W_{i,s_2}, \ldots $ in (\ref{folge1}) 
are greater than 0 (because of $s_1,s_2,\ldots \geq r_i$), and therefore it follows from the definition of $ a(j,s)$ that the $ k$-th oracle machine on input $ k$ halts for an infinite number of subsets of $\bbba$. The reason is that we have ${\mathcal M}_{ k}^{{\bbba}_{s_1}}({ k})\downarrow ^{s_1}$, ${\mathcal M}_{ k}^{{\bbba}_{s_2}}({ k})\downarrow ^{s_2}, \, \ldots$. By $(N_{ k})$ this means that 
${\mathcal M}_{ k}^{\bbba}({ k})\downarrow $ and, consequently, ${\mathcal M}_{ k}^{{\bbba}}( k)\downarrow ^t$ for some $t$.
This implies that there is a
 $t'$ such that any query of ${\mathcal M}_{ k}^{{\bbba}}$ on $ k$ does not refer to $\bbba\setminus \bbba_{t'}$ and thus we have  ${ \mathcal M}_k^{\bbba_{t'}}( k)\downarrow^t$. Moreover, 
${ \mathcal M}_{ k}^{\bbba_{t'}}, \ldots,{ \mathcal M}_{ k}^{\bbba_{\rm max\{t',t\}} } $, and ${ \mathcal M}_{ k}^{\bbba_{\rm max\{t',t\}+1} }, \ldots$ perform the same computation on input $k$.
In this way we get 
${ \mathcal M}_{ k}^{\bbba_{\rm max\{t',t\}} } ( k)\downarrow^{\rm max\{t',t\}} $  and   $\{a( k,s_1), a( k,s_2), \ldots\}$ is  therefore finite. Thus, there is an $ l$ such that $a( k,s_{ l}) < \max (W_{i,s_{ l}})$ since $\max (W_{i,s_1}) < \max (W_{i,s_2}) < \cdots$ is properly increasing. 
But,  this contradicts (\ref{folge1}) and, thus, our assumption is wrong and  we must  have $ \bbba\cap W_{i} \not=\emptyset$. 
 \end{proof}

Thus, $\bbba$ has the usual properties of a low and a simple set.
Since the intersection of $\bbba$ and any infinite halting set
$W_i\subseteq \bbbn$ is not empty and the intersection of the
complement $\bbba^{ \rm c}$ and $\bbbn$ is an infinite set (with
$\bbba^{ \rm c}\cap \bbba=\emptyset$), $\bbba^{ \rm c}\cap \bbbn$ is
not ${\sf M}_ {\rm add}^{1,=}$-semi-decidable.
\begin{lem} $\bbba^{ \rm c}\cap \bbbn$ and $\bbba^{ \rm c} $ are not semi-decidable by a machine in ${\sf M}_ {\rm add}^{1,=}$.
\end{lem}

\begin{cor} $\bbba$ is not decidable by a machine in ${\sf M}_ {\rm add}^{1,=}$.
\end{cor}

Although, in analogy with the classical setting, we could syntactically prove the following lemma, we will give a semantic proof that shows in detail the possibility to transfer the construction idea to machines over an arbitrary structure of finite signature if there is an effectively enumerable and decidable set $\bar \bbbn$ of elements of the structure like $\bbbn$. 

\begin{lem}\label{DiedritteEig} $\bbbk^\bbba\preceq \bbbk^\emptyset$. 
\end{lem}
\begin{proof}
 Since $\bbbn$ is ${\sf M}_{\rm add}^{1}$-decidable, it is enough to show that $\bbbk^\bbba$ and $\bbbn \setminus \bbbk^\bbba$ are semi-decidable by a machine in ${\sf M}_{\rm add}^{1}( \bbbk^\emptyset)$. To simplify matters, we will give corresponding algorithms for inputs $k\in \bbbr$ in a short form where 
at the beginning
the machines enumerate all positive integers and compare any enumerated number with the input $k$. If the input $k$ is not a positive integer, then the enumeration of integers will not be stopped. 

\vspace{0.2cm}

\noindent a) $\bbbk^\bbba$ is semi-decidable by:

\begin{itemize}
\item{Input:} $k\in \bbbr$.
\item Put $n:=0$. 
\item Repeat $n:=n+1$ until $k=n$. 
\item For $k\in \bbbn _+$, simulate ${ \mathcal M}_k^{\bbba} \in {\sf M}_{\rm add}^{1}({\bbba})$ on $k$ where the simulation of the queries whether $x\in \bbba$ is done as follows: 

Compute $ K_{\!{\mathcal L}_x} $ for the following ${ \mathcal L}_x\in {\sf M}_{\rm add}^{1}$ and ask whether $ K_{\!{\mathcal L}_x}\in \bbbk^\emptyset$.

\begin{description}

\item[${ \mathcal L}_x$]

\,\, Enumerate the elements of $\bbba$ until an enumerated element is equal to $x$.
\end{description}
 \end{itemize}

\vspace{0.2cm}

\noindent b) By condition $(N_k)$ we know that ${ \mathcal M}_k^{{\bbba}}(k)\downarrow $ holds if for every $i$ there is a $j> i$ such that ${ \mathcal M}_k^{{\bbba}_j}(k)\downarrow ^j $. Moreover, as  shown in the proof of Lemma 
\ref{ConditionsSimplicity}, ${ \mathcal M}_k^{{\bbba}}(k)\downarrow $ holds only if  there is an  $s$ such that  ${ \mathcal M}_k^{{\bbba}_j}(k)\downarrow ^j $ holds for all   $j> s$. 
Consequently, ${ \mathcal M}_k^{{\bbba}}(k)\uparrow $ holds if and only if  there is an $i$ such that the following machine ${ \mathcal L}_i^{(k)}\in {\sf M}_{\rm add}^{1}$ does not halt and thus $ K_{\!{\mathcal L}_i^{(k)}}\not \in \bbbk^\emptyset$.

\begin{description}
\item[${ \mathcal L}_i^{(k)}$]

\hspace{0.1cm} Input: $x\in \bbbr$.

\quad Put $j:=i$. 

\quad Repeat (1), (2), and (3) until an output instruction is simulated 

\quad\quad\quad \quad \quad \quad (that means until ${ \mathcal M}_k^{{\bbba}_j}(k)\downarrow ^j $ for some $j>i$).

\quad (1) Compute the elements $x_{i_1},\ldots, x_{i_{j}}$ of the set $\bbba$. 

\quad (2) Increment $j$ by 1. 

\quad (3) Simulate $j$ steps of the execution of ${ \mathcal M}_k^{{\bbba}_j} $ on $k$. 
 \end{description}
 ${ \mathcal L}_i^{(k)}\in {\sf M}_{\rm add}^{1}$ halts if and only if ${ \mathcal M}_k^{{\bbba}_j}(k)\downarrow ^j $ for some $j> i$. This means that 
there is an $i$ for that ${ \mathcal L}_i^{(k)}$ does not halt and thus $ K_{\!{\mathcal L}_i^{(k)}}\not \in \bbbk^\emptyset$ holds if and only if 
${ \mathcal M}_k^{{\bbba}}(k)\uparrow $ and, thus, $k\not \in \bbbk^\bbba$  are valid.

\vspace{0.1cm}

\noindent Consequently, $\bbbn\setminus \bbbk^\bbba$ is semi-decidable by:

\begin{itemize}
\item {Input:} $k\in \bbbr $.

\item Put $n:=0$. 

\item Repeat $n:=n+1$ until $k=n$. 

\item For $k\in \bbbn_{ +}$,  enumerate the positive integers $i$, compute, for every $i$, $ K_{\!{\mathcal L}_i^{(k)}} $ and ask whether $ K_{\!{\mathcal L}_i^{(k)}}\in \bbbk^\emptyset$. Halt if the answer is no.\qedhere
\end{itemize} 
\end{proof}

\noindent Consequently, we get the following.
\begin{prop}\label{AnichtaufK} $\bbba\mbox{$\strpreceq$} \bbbk^\emptyset$. \end{prop}

 \begin{proof} 
 $\bbba\preceq \bbbk^\emptyset$ follows from Corollary \ref{ReduktAequv} and Lemma \ref{ErsteEigenschaft}.
Let us assume ${\bbbk^\emptyset } \preceq \bbba$. Then, we would have $\bbbk^{\bbbk^\emptyset } \preceq \bbbk^\bbba \preceq \bbbk^\emptyset$ by Lemma \ref{DiedritteEig} and the halting problem $\bbbk^{\bbbk^\emptyset }$ would be decidable by a machine in ${\sf M}_{\rm add}^1(\bbbk^\emptyset)$ in contradiction to Lemma \ref{unentOrakelm}.
\end{proof}

Let us remark that 
the algorithms given in the proof of Lemma \ref{DiedritteEig} are sufficient to show $\bbbk^\bbba\strpreceq\!\!_{\rm add}^{1,=}\,\,\, \bbbh_{\rm add}^{1,=} $ ($\bbbk^\bbba\preceq_ {\rm add}^{1,=}\, \bbbh_{\rm add}^{1,=} $ holds since $\bbbn \preceq_{\rm add}^{1,=}\, \bbbh_{\rm add}^{1,=} $   and
$\bbbh_{\rm add}^{1,=}  \not \preceq_ {\rm add}^{1,=}\,  \bbbk^\bbba$ follows from  
 $\bbbn \not \preceq_{\rm add}^{1,=}\, \bbbk^\bbba $), but they are not sufficient to show $\bbbk^\bbba\preceq_{\rm add}^{1,=} \bbbk^\emptyset$  (since  $\bbbn \not \preceq_{\rm add}^{1,=}\, \bbbk^\emptyset $).

Moreover, the construction of $\bbba$  and the proofs show clearly that a similar construction  can be done 
 for the BSS RAM's over an arbitrary algebraic structure   of   finite signature if there is an infinite  set that is  enumerable and decidable over  this structure  like the set $\bbbn$. The 	restriction to a  finite number of operations, constants, and relations  allows  the  enumeration   and the simulation of  all BSS RAM's over such a structure by a universal machine and the  definition of  undecidable semi-decidable problems below the corresponding halting problem in the same way as above.  As the following example shows, the transfer is also  possible, even if $\bbbk^{ \emptyset} = \bbbh_{\rm spec}({\sf M}_{\rm add}^{1}(\emptyset))$   and, consequently,  the   Halting Problem    for Turing machines are decidable over  the considered  structure.  Let us consider the BSS RAM's  over  $\bbbr^{1,r_1}_{\rm add}=_{\rm df}(\bbbr;0,1,r_1;+,-;\geq)$ where the digits $\alpha_1, \alpha_2,\ldots \in  \{0,1\}$ of the real number $r_1=\sum_{j=1}^{\infty} \alpha_j 10^{-j}$ are defined by
 $\alpha_j= 1$  if $j\in  \bbbk^{ \emptyset} $ and  by  $\alpha_j=0 $ otherwise. Then, on one hand,  $\bbbk^{ \emptyset}$ is decidable by means of $r_1$ and, on the other hand,
 $\bbbr^{1,r_1}_{\rm add}$  allows to apply the priority method for constructing  new semi-decidable sets  since it is a structure of finite signature containing $\bbbn$.

\vspace{0.2cm}

 In order to extend the hierarchy (\ref{ErsteHierarchie}), we want to use that, for all inputs in $\bbbz^\infty$, problems such as $\bbbl\cap \bbbz^\infty\subseteq \bbbr^\infty$ are decidable in the classical sense. 
 More precisely, the decidability of $\bbbl\subseteq \bbbr^\infty$ for inputs in $\bbbz^\infty$ means that $\bbbl \cap \bbbz^\infty$ is ${\sf M}_{\rm add}^{1}$-decidable and, thus, the oracle $\bbbl$ has no benefit for the computation of additive machines in ${\sf M}_{\rm add}^{1}(\bbbl)$ on the input space $\bbbz^\infty$.
 This implies that the next lemma allows to transfer and generalize several results from the recursion theory and to show the incomparability of certain problems. 
\begin{lem} Let $ \mathcal Q$ be ${\sf M}_{\rm add}^{1}$-decidable for all inputs in $\bbbz^\infty$ and $P\subseteq \bbbn$. Then, we have $P \not\preceq \mathcal Q$ if $(\bbbr^\infty\setminus P)\cap \bbbn$ is not semi-decidable by a machine in ${\sf M}_{\rm add}^{1,=}$. 
\end{lem}

\begin{proof} Let us assume that $P$ is decidable by a machine in $ {\sf M}_{\rm add}^{1}( \mathcal Q)$.
Consequently, $\bbbr\setminus P$ is semi-decidable by a machine ${ \mathcal M} \in {\sf M}_{\rm add}^{1}( \mathcal Q)$. 
Let us modify ${ \mathcal M} $ such that at the beginning
the machine ${ \mathcal M} $ enumerates all non-negative integers and compares any enumerated number with the input. 
If the input is a non-negative integer, then the further instructions of ${ \mathcal M}$ can be simulated by a machine in ${\sf M}_{\rm add}^{1,=} $ since each order test can be simulated by means of few equality tests and all queries of ${ \mathcal M}$ can be replaced by a decision procedure.
Thus $(\bbbr\setminus P)\cap \bbbn$ is semi-decidable by a machine in ${\sf M}_{\rm add}^{1,=}$.
Consequently, if $(\bbbr^\infty\setminus P)\cap \bbbn$ is not semi-decidable by a machine in ${\sf M}_{\rm add}^{1,=}$, then the assumption is wrong. 
\end{proof}

\begin{prop}\label{Uebertr} Let $ \mathcal Q$ be decidable for all inputs in $\bbbz^\infty$ and $P\subseteq \bbbn$. Then, we have $P \not\preceq \mathcal Q$ if $P$ or $(\bbbr^\infty\setminus P)\cap \bbbn$ is not semi-decidable by a machine in ${\sf M}_{\rm add}^{1,=}$.
\end{prop}

Therefore, we can extend the hierarchy (\ref{ErsteHierarchie}) since we have $P \not\preceq \bbbl$ for any $P\subseteq \bbbn$ satisfying the conditions in accordance with Proposition \ref{Uebertr} 
 and, moreover,   we have 
 $ \bbbl\not\preceq P$ for any $P\subseteq \bbbn$.
 The latter holds since
$ \bbbl_2\not\preceq P$ 
can be shown in analogy with 
 $ \bbbl_2\not\preceq \bbbq$ (cf. \cite{4}) and, furthermore, $ \bbbl_1\not\preceq P$ 
can be shown in a similar way  by proving  that, for  any machine ${\mathcal M}\in {\sf M}_{\rm add}^{1}(P)$  that  decides a  problem, there is a rational number $q$ such that  ${\mathcal M}$ goes through the same computation path for  both inputs $q$ and  $\pi$. (For more information about this kind of proofs   see  also      the proof of Lemma \ref{LLL9}.)

Hence, for any problem $P\subseteq \bbbn$ that is ${\sf M}_{\rm add}^{1,=}$-semi-decidable, $ \bbbl$ is incomparable with  $P$  if $(\bbbr^\infty\setminus P)\cap \bbbn$ is not semi-decidable by a machine in ${\sf M}_{\rm add}^{1,=}$. Under consideration  of  Proposition \ref{Turingred}
and Corollary \ref{ReduktAequv} we get the following.  

\begin{prop}\label{AKunverglL} $ \bbbl$ is incomparable with the Halting Problem for Turing machines and the problems $\bbbk^\emptyset$ and $\bbba$ with respect to the strong Turing reduction $\preceq$.\end{prop}

\begin{cor} $\bbbk^\emptyset\mbox{$\strpreceq$} \bbbh_{\rm add}^{1,=}$ and $\bbba\mbox{$\strpreceq$} \bbbh_{\rm add}^{1,=}$. \end{cor} 

 Since the positive integers can be unary encoded, for any $P\subseteq \bbbn$ we also want to consider the decision problem
\[\bbbo_P=_{\rm df}\bigcup_{n\in P} \bbbr^n =\bigcup_{n=1}^\infty\{(x_1,\ldots,x_n)\in\bbbr^n\mid n\in P\}.\]
Then, we have $\bbbo_\bbba\not\preceq \bbbl $ and $ \bbbl \not\preceq \bbbo_\bbba $. On the other hand, for any $n\in \bbbn$, $\bbbl_n$ is decidable by an additive oracle machine using the oracle set 
 $\bbbo_P\cap \bbbl=\bigcup_{k\in P} \bbbl_k$ if $P\subseteq\bbbn$ is infinite, ${\sf M}_{\rm add}^{1,=}$-semi-decidable, and thus ${\sf M}_{\rm add}^{1,=}$-enumerable. This implies the following by Proposition \ref{AnichtaufK} and Proposition \ref{AKunverglL}.
\begin{thm} We have

\vspace{0.2cm}

$ \bbbl \mbox{$\strpreceq$} \bigcup_{k\in \bbba} \bbbl_k \mbox{$\strpreceq$} \bigcup_{k\in \bbbk^ \emptyset} \bbbl_k \preceq \bbbh_{\rm add}^{1,=}$.
\end {thm} 

\section{ A hierarchy between $\bbbh_{\rm add}^{1,=}$ and $\bbbh_{\rm add}^{1}$ }
The aim of this section is to show $\bbbh_{\rm add}^{1,=}\mbox{$\strpreceq$} \bbbh_{\rm add}^{1}$ and to define an infinite hierarchy between $\bbbh_{\rm add}^{1,=}$ and $\bbbh_{\rm add}^{1}$ by using the properties of prime numbers and their square roots and the fact that the set $\bbbr \setminus \{r\}$ is semi-decidable by a machine in ${\sf M}_{\rm add}^{1}$ if $r$ is Turing computable. That means, we can use that, for every prime number $p$, the set $\bbbr \setminus \{\sqrt p \}$ is  ${\sf M}_{\rm add}^{1}$-semi-decidable. On the other hand, 
$\bbbr \setminus \{\sqrt p \}$ is not ${\sf M}_{\rm add}^{1}$-decidable  
and
we can even  show that $\bbbr \setminus \{\sqrt p \}$ cannot be BSS reduced to $\bbbh_{\rm add}^{1,=}$ with respect to $\preceq$. Before we start the construction of a hierarchy between $\bbbh_{\rm add}^{1,=}$ and $\bbbh_{\rm add}^{1}$, let us mention that the latter fact implies also that
 the sets $\bbbr \setminus \{\sqrt p \}$ and $\bbbh_{\rm add}^{1,=}$ define incomparable strong Turing degrees since $ \bbbl_2\not\preceq \bbbr \setminus \{\sqrt p \}$ can also be shown in analogy with $ \bbbl_2\not\preceq \bbbq$.

Let $p_1=2,p_2=3,\ldots$ be an enumeration of the prime numbers and let ${ \mathcal K}$ be an additive machine in ${\sf M}_{\rm add}^{1}$ that semi-decides $\bigcup_{i=1}^\infty \{i\}\times (\bbbr\setminus \{\sqrt{p_i} \})$ by checking, for any input $(i,x)$, the condition ($x< \frac{r}{q}$ and $\frac{r^2}{q^2}< p_i$) or ($x> \frac{r}{q}$ and $\frac{r^2}{q^2}> p_i$) for all enumerated $(r,q)\in \bbbn\times \bbbn_+$ until the condition is satisfied. Note, that the machine ${ \mathcal K}$ can compute the prime number $p=p_j$ from $j$ by executing instructions realizing the algorithm $k:=1$; $p:=2$; while $k<j$ \{$p:=p+1$; if $(\forall (t,s)\in\{2,\ldots,p-1\}^2)(p\not=t\cdot s)$ then $k:=k+1$;\}.
Therefore,  we get
\[
\bbbp_i \,\,=_{\rm df} \,\,\{(2\,.\,(i, x)\,.\,{\rm code}({ \mathcal K})) \mid x\in \bbbr\setminus \{\sqrt{p_i} \}\} \subseteq \bbbh_{\rm add}^{1}
\]
 for $i\geq 1$  and, for the halting problems 
\[\hspace{0.5cm}
\bbbh_i\,\,=_{\rm df}\,\,\bbbh_{\rm add}^{1,=}\cup \bigcup_{j\leq i} \bbbp_j 
\] 
 (including  $\bbbh_0=\bbbh_{\rm add}^{1,=}$),  the following hierarchy. 
\begin{lem}$ \bbbh_{\rm add}^{1,=} = \bbbh_0 \preceq \bbbh_1 \preceq \bbbh_2\preceq \cdots\preceq \bbbh_{\rm add}^{1}$.\end{lem}

The following lemmas are useful for giving further incomparable Turing degrees with respect to $\preceq$.
\begin{lem}\label{L9} For any $1\leq k\leq i$, $\{\sqrt{p_k}\} \preceq
\{\sqrt{p_1},\ldots, \sqrt{p_i} \}$. \end{lem}
\begin{proof} Let us give, for the oracle ${\mathcal O}_i=\{\sqrt{p_1},\ldots, \sqrt{p_i} \}$, a machine ${ \mathcal L}$ in ${\sf M}_{\rm add}^1({\mathcal O}_i)$ 
 that outputs 1 if and only if the input is $\sqrt{p_k}$ and otherwise 0. At the beginning ${ \mathcal L}$ queries the oracle ${\mathcal O}_i$ whether the input $x$ belongs to this oracle and outputs 0 if the answer is in the negative. 
Otherwise,
 $x$ is in $ {\mathcal O}_i$ and
there is an index $j_0\leq i$ for which 
we have $x=\sqrt{p_{j_0}} $. In this case, the index $j_0$ is searched by means of    comparisons of $x^2$ with the squares of the elements of ${\mathcal O}_i$ and a procedure of excluding members from the 
set $\{1,\ldots, i\}$ as follows.

 Let $I:=\{1,\ldots, i\}$ be the set of the indices of the first $i$ prime numbers and let ${ \mathcal L}$ simulate the machine ${ \mathcal K} $ for the inputs $(1,x),\ldots (i,x)$ simultaneously. For all enumerated $(r,q)$, ${ \mathcal L}$ checks $x< \frac{r}{q} $ and $ \frac{r^2}{q^2}< p_j$ as well as $x> \frac{r}{q}$ and $ \frac{r^2}{q^2}> p_j$ for any $j\in I$. If one of both conditions is satisfied for some $j$ and, consequently, $x\not=\sqrt{p_j}$, then $j$ is deleted in $I$ by putting $I:=I\setminus \{j\}$. If $k=j$, then the machine halts and the output is 0. Otherwise, the machine continues its operations. If $I$ contains only one element, then this member of $I$ can be only the index $j_0=k$ and should not  be deleted,  and in this case the output is 1.
\end{proof}
 Since relationships as $\{\sqrt{p_1},\ldots, \sqrt{p_i} \} \equiv \bbbr \setminus \{\sqrt{p_1},\ldots, \sqrt{p_i}\} $
are characteristic for the Turing reducibility relation and $\{\sqrt{p_j}\} \equiv \bbbp_j $ is also easy to see, we even have the following.
 \begin{lem}\label{MengenPrimzahlen} For any $i\geq 1$, 
 $
\{\sqrt{p_1},\ldots, \sqrt{p_i} \} \equiv \bigcup_{j\leq i} \bbbp_j $. \end{lem}
\begin{proof} 1. $\bbbr \setminus \{\sqrt{p_1},\ldots, \sqrt{p_i}\}\preceq \bigcup_{j\leq i} \bbbp_j $ 
 holds since $x \in \bbbr \setminus \{\sqrt{p_1},\ldots, \sqrt{p_i}\}$ is satisfied 
 if and only if $\{(2\,.\, (1,x)\,.\, {\rm code}({ \mathcal K})),\ldots, (2\,.\, (i,x)\,.\, {\rm code}({ \mathcal K}))\} \subset \bigcup_{j\leq i} \bbbp_j$ holds. 2. For any input $(2\,.\, (j,x)\,.\, {\rm code}({ \mathcal K}))$ with  $j\leq i$, we can check whether $x\not =\sqrt{p_j}$ by using the oracle $\{(1,\sqrt{p_1}),\ldots, (i,\sqrt{p_i})\} $. By Lemma \ref{L9} $\{(1,\sqrt{p_1}),\ldots, (i,\sqrt{p_i})\} $ can be reduced to $\{\sqrt{p_1},\ldots, \sqrt{p_i} \}$.
 Thus, we also obtain $\bigcup_{j\leq i} \bbbp_j \preceq \{\sqrt{p_1},\ldots, \sqrt{p_i} \}$.
\end{proof}
 Therefore, the decidability of the set $\{{\rm code}({ \mathcal K})\}$ implies also the following.

\begin{lem}\label{LL9} Let $i\geq 1$ and ${ \mathcal O}\subseteq  \bbbr^\infty$. If $\,\bbbh_i$ is decidable by an oracle machine in ${\sf M}_{\rm add}^{1} ({ \mathcal O}) $, then the problems $\bigcup_{j\leq i} \bbbp_j $ and $\{\sqrt{p_1},\ldots, \sqrt{p_i}\}$ are also decidable by some machines in ${\sf M}_{\rm add}^{1} ({ \mathcal O}) $. 
\end{lem}
Together with Lemma \ref{L9} this means that the problem $ \bbbp_{i+1}$ as well as the problem $\{ \sqrt{p_{i+1}}\}$ would be decidable by machines in ${\sf M}_{\rm add}^{1} (\bbbh_i) $ if $\bbbh_{i+1}$ would be decidable by a machine in ${\sf M}_{\rm add}^{1} (\bbbh_i) $. In the following we will use it.

\begin{cor}\label{Zusammenh} For any $i\geq 0$, if $\{ \sqrt{p_{i+1}}\}\not \preceq\bbbh_i$, then $\bbbh_{i+1}\not\preceq\bbbh_i $. 
\end{cor}

First, let us consider the machines in $ {\sf M}_{\rm add}^{1} (\bbbh_{\rm add}^{1,=}) $ and in this context, in particular, queries for inputs in $\bbbr$. Recall that, for inputs $x\in \bbbr$, at any time of the computation the content of a register of these machines
can be described by a term $kx+l$ with integers $k$ and $l$ that have the following features.
$k$ and $l$ are integers that result only from the execution of the computation instructions along the traversed path, and, for any fixed computation path, $k$ and $l$ are independent of the input values. In the following, we will consider only this type of terms and we will use that,
 for any fixed computation path and inputs $x\in \bbbr$, any oracle query coincides with a question
 of the form \begin{equation}\label{OrakelAnfr0}(k_0 x+l_0,k_1 x+l_1,\ldots, k_m x+l_m) \in \bbbh_{\rm add}^{1,=}\mbox{\em?}\end{equation} where $k_i,l_i\in \bbbz$ ($i=0,\ldots,m$) are independent of the input values in the described sense. 
Moreover, we use that, in case of $k\not=0$, we have $kx+l\in \bbbr\setminus \bbbq$ if and only if $x\in \bbbr\setminus \bbbq$.

\begin{lem}\label{OrakelWurzUndPi} Let ${\mathcal M}$ be a machine in $ {\sf M}_{\rm add}^{1} (\bbbh_{\rm add}^{1,=}) $. 

1. For all inputs $x\in \bbbr$, the oracle queries $\vec z \in \bbbh_{\rm add}^{1,=}\mbox{\em ?} $ of ${\mathcal M} $ which are not of the form \begin{equation}\label{OrakelAnfr}(j\,.\,( k_1 x+l_1,\ldots, k_j x+l_j)\,.\, {\rm code}({ \mathcal N}))\in \bbbh_{\rm add}^{1,=}\mbox{\em?}\end{equation} 
 for some $j, k_i,l_i\in \bbbz$ and ${\mathcal N}\in {\sf M}_{\rm add}^{1,=} $ are always answered in the negative. 

 2. If, for some input $x_0\in \bbbr\setminus \bbbq$, an oracle query of  the form (\ref{OrakelAnfr0}) coincides  with   \begin{equation}\label{Queryforx0}(j\,.\,( k_1 x_0+l_1,\ldots, k_j x_0+l_j)\,.\, {\rm code}({ \mathcal N}))\in \bbbh_{\rm add}^{1,=}\mbox{\em?}\end{equation} for some $j, k_i,l_i\in \bbbz$ and ${\mathcal N}\in {\sf M}_{\rm add}^{1,=}$, then this oracle query has the same form (\ref{OrakelAnfr}) (with the same $k_i,l_i\in \bbbz$) for all inputs $x\in \bbbr$ for which ${\mathcal M}$ goes through the same computation path as for $x_0$ until  this query  is executed.
\end{lem} 
\begin{proof} Let $B_{x_0}$ be the computation path traversed by ${\mathcal M}$ on $x_0\in \bbbr\setminus \bbbq$. Then, for  inputs $x\in \bbbr$, any query of $B_{x_0}$ has the general form
\begin{equation}\label{Queryforx} (k_0x+l_0\,.\,( k_1 x+l_1,\ldots, k_j x+l_{j})\,.\,(
\underbrace{ k_{j+1} x+l_{j+1},\ldots, k_{j+r} x+l_{j+r}}_{=\vec y}))\in \bbbh_{\rm add}^{1,=}\mbox{\em ?}.
\end{equation}
 Let us assume that a query $Q$ of the form (\ref{Queryforx}) corresponds, for the input $x_0$, to (\ref{Queryforx0}). Then, we have $k_0x_0+l_0=j$ and there is a machine ${\mathcal N}\in {\sf M}_{\rm add}^{1,=} $ such that $\vec y$ is its code ${\rm code}({\mathcal N})\in \{0,1\}^\infty$. 
Since $x_0$ satisfies equations of the form $f(x)=_{\rm df} k_0 x+l_0=j$, $g_\nu(x)=_{\rm df} k_\nu x+l_\nu=0$, and $h_\mu(x)=_{\rm df} k_\mu x+l_\mu=1$ only if $k_0= k_\nu=k_\mu=0$ holds, the number $j$ and ${\rm code}({\mathcal N})$ are determined by constant functions $f(x)=j$, $g_\nu(x)=0$, and $h_\mu(x)=1$ and thus, for all $x\in\bbbr$ for which ${\mathcal M}$ goes through the computation path $B_{x_0}$, the query $Q$ has the same form (\ref{OrakelAnfr}).
\end{proof}
\begin{lem}\label{LLL9} Let ${ \mathcal M}\in {\sf M}_{\rm add}^{1} (\bbbh_{\rm add}^{1,=}) $ be a machine deciding a problem $S\subseteq \bbbr$. Then, there are $n,m\in \bbbn_{ +}$ such that ${ \mathcal M}$ rejects the inputs $\sqrt{2}$ and $\frac{n}{m}\pi$ or ${ \mathcal M}$ accepts both inputs.
\end{lem}
\begin{proof}By Lemma \ref {OrakelWurzUndPi}, for any computation path $B$ of an ${ \mathcal M}\in {\sf M}_{\rm add}^{1} (\bbbh_{\rm add}^{1,=}) $, there is a system $S_B$ of conditions of the form
\begin{eqnarray}
k_\nu x+l_\nu\,\, \geq\,\, 0 \quad \mbox{ and } \quad k_\mu x+l_\mu&>& 0,\label{System1}\\
(j\,.\, (k_1 x+l_1,\ldots, k_j x+l_{j})\,.\, {\rm code}({\mathcal N}))&\in& \bbbh_{\rm add}^{1,=},\label{System2} \\
(j\,.\, (k_1 x+l_1,\ldots, k_j x+l_{j})\,.\,{\rm code}({\mathcal N}))&\not\in& \bbbh_{\rm add}^{1,=}\label{System3}
\end{eqnarray}
($k_i,l_i\in \bbbz$, ${ \mathcal N}\in {\sf M}_{\rm add}^{1,=}
$) that is satisfied by an input $x\in \bbbr\setminus \bbbq$ if and only if this path $B$ is traversed by ${ \mathcal M}$ on $x$.

Moreover, every computation path of a machine  ${ \mathcal N}$ in ${\sf M}_{\rm add}^{1,=} $ on an input of the form $ ( k_1 x+l_1,\ldots, k_j x+l_{j})$ can be described by equations and inequalities of the form 
\begin{displaymath}\begin{array}{c}
k_\nu x+l_\nu\,=\,0 \quad \mbox{ and } \quad 
k_\mu x+l_\mu\,\not =\, 0. 
\end{array}\end{displaymath}
Since, for any $n,m\in \bbbn_{ +}$, $\frac{n}{m}\pi$ and $\sqrt{2}$ satisfy the same equations of the form $k_\nu x+l_\nu=0$, every ${ \mathcal N}\in {\sf M}_{\rm add}^{1,=} $ halts on $ (k_1 \frac{n}{m}\pi +l_1,\ldots, k_j \frac{n}{m}\pi +l_j)$ if and only if it halts on $ (k_1 \sqrt{2} +l_1,\ldots, k_j \sqrt{2} +l_j)$. Moreover, since the computation path $B_ {\sqrt{2}}$ of ${\mathcal M}$ on input $\sqrt{2}$ is finite,  the system $S_{B_{\sqrt{2}}}$ contains 
only a finite number of inequalities of the form (\ref{System1}) and, thus,  these   inequalities  can be described  by  some $(k_\nu, l_\nu), (k_\mu, l_\mu) \in K\times L $  where $K\subseteq \bbbz\setminus \{0\} $  and  $L\subseteq \bbbz$ are  finite sets. From this it follows that there is a real number  $\varepsilon$ with $k \sqrt{2}+l>\varepsilon>0$ for all $(k,l)\in K\times L $. Moreover, we also have integers $n_\varepsilon,m_\varepsilon \in \bbbn_{+}$ with $|\frac{n_\varepsilon }{m_\varepsilon }-\frac{\sqrt{2}}{\pi} |<\frac{\varepsilon }{|k\pi|} $ for $k\in K$ such that $|k \frac{n_\varepsilon}{m_\varepsilon}\pi -k {\sqrt{2}}|< \varepsilon$ and, consequently, $k {\sqrt{2}}-k\frac{n_\varepsilon}{m_\varepsilon}\pi < \varepsilon$ and
 \[k \frac{n_\varepsilon}{m_\varepsilon}\pi+l=
k\sqrt{2} +l +k \frac{n_\varepsilon}{m_\varepsilon}\pi -k {\sqrt{2}}=\underbrace{k \sqrt{2} +l }_{>\varepsilon}-(\underbrace{k {\sqrt{2}}-k \frac{n_\varepsilon}{m_\varepsilon}\pi }_{< \varepsilon}) >0\] hold for all  $(k,l)\in K\times L $.
That means that
$\frac{n_\varepsilon}{m_\varepsilon}\pi$ and $\sqrt{2}$ satisfy the same system $S_{B_ {\sqrt{2}}}$ and $B_ {\sqrt{2}}$ is traversed by ${\mathcal M}$ for both inputs $\sqrt{2}$ and $\frac{n_\varepsilon}{m_\varepsilon}\pi$.
\end{proof}

Thus, we get the following corollaries.
\begin{cor} The problem $\{\sqrt{2}\}$ is not decidable by a machine in ${\sf M}_{\rm add}^{1} (\bbbh_{\rm add}^{1,=})$.\end{cor}
\begin{cor} 
$\bbbh_{\rm add}^{1,=} \strpreceq \bbbh_1 $.\end{cor}

Our next goal is to show that $\bbbh_i$ is strictly easier than $\bbbh_{i+1}$ for any $i\geq 1$ by using Corollary \ref{Zusammenh}.

\begin{lem}\label{Wurzeltest} Let ${ \mathcal M}\in {\sf M}_{\rm add}^{1} (\bbbh_i) $ be a machine deciding a problem $S\subseteq \bbbr$. Then, there are $n,m\in \bbbn_{ +}$ such that ${ \mathcal M}$ rejects $\sqrt{p_{i+1}}$ and $\frac{n}{m}\pi$ or ${ \mathcal M}$ accepts the both inputs.
\end{lem}
\begin{proof} In analogy with the proof of Lemma \ref{LLL9}, we can show  that   for any  finite system of  conditions of the form (\ref{System1}), (\ref{System2}), and (\ref{System3}), there are $n,m\in \bbbn_{ +}$ such that both $\sqrt{p_{i+1}}$ and $\frac{n}{m}\pi$ satisfy this system. Moreover, we have $k_2\sqrt{p_{i+1}}+l_2\not=\sqrt{p_{j}}$ and $k_2\frac{n}{m}\pi+ l_2\not=\sqrt{p_{j}}$ for any $k_2\not=0$, $j\leq i$, and $m,n\in \bbbn_{ +}$. Hence, for all $j\leq i$,  questions of the form $ (2\,.\,(j,k_2 \frac{n}{m}\pi+ l_2)\,.\,{\rm code}({ \mathcal K})) \in \bbbh_i\mbox{\em?}
$ and $ (2\,.\,(j,k_2 \sqrt{p_{i+1}}+ l_2)\,.\,{\rm code}({ \mathcal K})) \in \bbbh_i\mbox{\em?}
$  are answered in the positive.
\end{proof}

From this lemma we deduce that $\{\sqrt{p_{i+1}}\} $ is not decidable by a machine in ${\sf M}_{\rm add}^{1} (\bbbh_i)$.
\begin{cor} For any $i\geq 1$, $\{\sqrt{p_{i+1}}\}\not \preceq \bbbh_i$.\end{cor}
Moreover, the proof of Lemma \ref{Wurzeltest} is also suitable for showing that $\{\sqrt{p_{i+1}}\} $ is not decidable by machines using only $\bigcup _{j\leq i}\bbbp_j$  as oracle. Thus, by Lemma \ref{MengenPrimzahlen} we get the following.

\begin{prop}
 For any $i\geq 1$, $\{\sqrt{p_{i+1}}\}\not \preceq \bigcup _{j\leq i}\bbbp_j$ and $\{\sqrt{p_{i+1}}\}\not \preceq \{\sqrt{p_{1}},\ldots, \sqrt{p_{i}}\}$.\end{prop}

Since these results are independent of the order of the prime numbers, we get incomparable Turing degrees  with respect to $\preceq$.

\begin{prop}
 For any $i,j\geq 1$ where $i\not =j$, we have $\bbbp_i  \not\preceq \bbbp_j$ and $\bbbp_j \not \preceq \bbbp_i$. \end{prop}

Moreover, for $i\geq 1$, we have  $\bbbp_i  \not\preceq  \bbbh_{\rm add}^{1,=}$ and thus $\bbbp_i  \not\preceq  P$ for any $P \preceq  \bbbh_{\rm add}^{1,=}$. Therefore,  for    $P=\bbbk^\emptyset$ and $P= \bbba $,  the problems  $\bbbp_i $  and $P$   are     incomparable   by Proposition \ref{Uebertr} and, since $\bbbq \not\preceq \{\sqrt{p_{i}}\}$ is easy to show,  $\bbbp_i $  and  $\bbbq$     as well as
  $\bbbp_i $ and  $\bbbh_{\rm add}^{1,=}$     are     incomparable.

By Corollary \ref{Zusammenh} we get the following relation.

\begin{lem} For any $i\geq 0$,
$\bbbh_i
 \strpreceq \bbbh_{i+1}$.\end{lem}
Consequently, we have the following. 
\begin{thm} Let $k\in \bbbn$.

\vspace{0.2cm}

$ \bbbh_{\rm add}^{1,=} \strpreceq \bbbh_1 \strpreceq\cdots \strpreceq \bbbh_k\strpreceq \cdots \strpreceq \bigcup_{i\geq 1} \bbbh_i \preceq \bbbh_{\rm add}^{1}$.\end{thm}

\section{Summary}

We have defined a number of problems that are not decidable by additive BSS machines without irrational numbers in order to give an initial characterization of the strong Turing reduction by additive machines and some resulting Turing degrees below the halting problem $ \bbbh_{\rm add}^{1}$.

\begin{figure}[ht]
\label{HalteprZusammenf}
 \begin{displaymath}\begin{array}{lllllllllll}

& \bbbh_{\rm add}^{1}& \equiv\,\,\,\, \bbbh_{\rm add}\\
&\,\,\,\,\vdots   &\hspace{3.7cm}      \vdots\\

&\,\,\,\uparrow  \\
&\,\bbbh_2 &=\bbbh_{\rm add}^{1,=}\cup \bbbp_1 \cup \bbbp_2    \hspace{0.18cm}     \leftarrow   \bbbp_2\\ 
&\,\,\,\uparrow \\

&\,\bbbh_1 &=\bbbh_{\rm add}^{1,=}\cup \bbbp_1    \hspace{1cm}       \leftarrow  \bbbp_1\\ 
&\,\,\,\uparrow \\

&\,\bbbh_{\rm add}^{1,=}& \leftarrow \,\, \bigcup_{k\in \bbba} \bbbl_k \leftarrow \bbbl \cdots \leftarrow \bbbl_5\leftarrow \bbbl_4\leftarrow \bbbl_3\leftarrow \bbbl_2\leftarrow\bbbl_1 \,\,\,\, 
 \equiv\!\!& \bbbq \qquad \qquad \qquad \qquad \qquad \qquad\qquad \qquad \qquad \quad\quad \quad\vspace{0.2cm}\\

&\,\,\,\uparrow &\hspace{1.3cm}\uparrow \\

\,\,\,\,&\,\,\bbbk^\emptyset &\leftarrow \,\,\,\,\,\,\,\,\,\, \,\,\bbba \\

\end{array}\end{displaymath}
\caption{Hierarchies below $ \bbbh_{\rm add}^{1}$}\label{Zusammenfassung}
\end{figure}

In Figure \ref{Zusammenfassung} we summarize the results of this paper. The arrow $\rightarrow$ stands for the relation $ \strpreceq\!\! $ (that means for $ \strpreceq\!\!_{\rm add}^{1} $) and $ \equiv$ stands for the equivalence relation $ \equiv_{\rm add}^{1} $. The picture is complete up to transitivity.
Two problems are incomparable   with respect to $ \preceq_{\rm add}^{1} $ if there is no directed sequence of arrows from one problem to  the other.

\section{Acknowledgement}
 I would like to thank
 the anonymous referees for their valuable comments and suggestions to improve the paper. Moreover, my thanks go to the 	participants of the meeting ''Real Computation and BSS Complexity'' in Greifswald for the discussion. In particular, I would like to thank Paul Grieger and Arno Pauly for very useful hints.

\end{document}